\documentclass{llncs}

\usepackage{color}
\usepackage{times, amssymb, amsmath, comment}
\usepackage{hyperref}
\usepackage[boxed,linesnumbered]{algorithm2e}
\usepackage{verbatim}
\newcommand{\real}{\mbox{${\mathbb R}$}}
\usepackage{amssymb}
\usepackage{textcomp}

\mainmatter

\newtheorem{conj}[theorem]{Conjecture}

\newtheorem{obs}[theorem]{Observation}
\newtheorem{lem}[theorem]{Lemma}

\renewenvironment{proof}[1][]{\begin{trivlist}
\item[\hspace{\labelsep}{\bf\noindent Proof#1:\/}] }{\qed\end{trivlist}}

\begin{document}

\pagestyle{headings}
\title{Helly-Type Theorems in Property Testing}
\author{ Sourav Chakraborty\inst{1} \and
         Rameshwar Pratap\inst{1} \and
         Sasanka Roy\inst{1} \and
         Shubhangi Saraf\inst{2}
       }

\authorrunning{Chakraborty, Pratap, Roy, Saraf}

\tocauthor{
Sourav Chakraborty (CMI, Chennai, India),
Rameshwar Pratap (CMI, Chennai, India),
Sasanka Roy (CMI, Chennai, India) and
Shubhangi Saraf (Rutgers University, USA)
}

\institute{
Chennai Mathematical Institute,\\
Chennai, India.\\
e-mail:\{sourav,rameshwar,sasanka\}@cmi.ac.in
\and
Department of Mathematics and Department of Computer Science,\\
Rutgers University.\\
e-mail:{shubhangi.saraf}@rutgers.edu
}

\maketitle

\begin{abstract}
Helly's theorem is a fundamental result in discrete geometry, describing the ways in which convex 
sets  intersect with each other. If $S$ is a set of $n$ points in $\real^d$, we say that $S$ 
is $(k,G)$-clusterable if it can be partitioned into $k$ clusters (subsets)  such that each cluster 
can be contained in a translated copy of a  geometric object $G$. In this paper, as an application of Helly's 
theorem, by taking a constant size sample from $S$, we present a testing algorithm for 
$(k,G)$-clustering, \textit{i.e.}, to distinguish between two cases: when $S$ is $(k,G)$-clusterable,
and when it is $\epsilon$-far from being $(k,G)$-clusterable. A set $S$ is $\epsilon$-far 
$(0<\epsilon\leq1)$ from being $(k,G)$-clusterable if at least $\epsilon n$ points need to be 
removed from $S$ to make it $(k,G)$-clusterable. We  solve this problem for $k=1$ and 
 when $G$ is a symmetric convex object. For $k>1$, we  solve a \textit{weaker} version of 
 this problem. Finally, as an application of our testing result, in clustering with outliers, 
 we show that one can find the \textit{approximate} clusters  by querying a constant 
 size sample, with high probability.
\end{abstract}


\section{Introduction}
Given a set of $n$ points in $\real^d$,  deciding whether all the
points can be contained in a unit radius ball is a well known problem
in Computational Geometry.   Of course, the goal is to solve this
problem as quickly as possible.  In order to solve this problem exactly, 
one has to look at all the $n$ points in the worst case scenario.  
But if $n$ is too large, an algorithm with linear running time
may not be fast enough. 
Thus, one may be interested in ``solving'' the
above problem  by taking a very small size sample and outputting the
``right answer'' with high probability. In this paper, we consider 
the \textit{promise} version of this problem.  
More precisely, for the given \textit{proximity parameter} $\epsilon~(\mbox{where}~0<\epsilon\leq1)$, our 
goal is to distinguish between the following two cases: 
\begin{itemize}
 \item  all the points can be contained in a unit radius ball,
 \item no unit radius ball can contain more than  $(1 - \epsilon)$ fraction of  points. 
\end{itemize}
The above \textit{promise} problem falls in the realm of property
testing (see~\cite{goldreich98}, ~\cite{goldreich97} and ~\cite{ron08}). 
In property testing, the goal is to  look at a very small fraction of the input and 
decide whether the input satisfies the property or is
``far'' from satisfying it. Property testing algorithms for computational
geometric problems have been studied earlier in \cite{czumaj1}, \cite{czumaj2} and
 \cite{alon}.   In this paper, we study the above problem in property testing setting and give a 
 simple algorithm to solve it.  The algorithm queries  only a constant number of 
 points (where the constant depends on the dimension $d$ and $\epsilon$, but is independent of $n$) and correctly distinguishes
between the  two cases mentioned above with probability at least $2/3$. While the algorithm
is very simple, the proof of correctness is a little involved, for which we use  Helly's theorem.
Helly's theorem (\cite{helly}) states that  if a family of convex sets in $\real^d$
has a non-empty intersection  for every  $d+1$  sets, then the whole family has a non-empty
intersection. In fact, since Helly's theorem also works for symmetric convex bodies,
we can solve the above problem for any symmetric convex body instead of just a
unit radius ball. Thus, we have 

 \begin{theorem}\label{thm:main}
Let $A$ be a symmetric convex body.  If $S$ is a set of $n$ points in
$\real^d$ as input with the proximity 
 parameter $\epsilon~ (\mbox{where~} 0<\epsilon\leq1)$, then there is an algorithm $\mathcal{A}$ that randomly
samples $O(\frac{d}{\epsilon^{d+1}})$ many points and \begin{itemize}
\item $\mathcal{A}$ accepts, if all the points in $S$ can be contained in a 
      translated copy of $A$,
\item $\mathcal{A}$  rejects with probability $\geq 2/3$, if any translated 
      copy of $A$ can contain at most $(1-\epsilon)n$  points.
\end{itemize}
The running time of $\mathcal{A}$ is $O(\frac{d}{\epsilon^{d+1}})$. 
\end{theorem}

One would  like to  generalize the above problem for more than one
object, \textit{i.e.},  given $k$ translated  copies of object
$B$, the goal is to distinguish between the following two cases with high
probability: 
\begin{itemize}
\item all $n$ points can be contained in $k$  translated copies of $B$, 
\item at least $\epsilon$ fraction of points cannot be contained in any $k$ translated copies of $B$.

\end{itemize}

We would like to conjecture that a similar algorithm, as stated in
Theorem~\ref{thm:main}, would also work for the generalized $k$ object
problem.  Unfortunately,  Helly's theorem does not hold for the $k$
object setting, but we would like to conjecture that  a version
of the Helly-type theorem  does hold for this setting.
Assuming the above conjecture, we can obtain a similar algorithm for
the $k$ object setting. We can also unconditionally solve a \textit{weaker} 
version of the $k$ object problem. 

\

\textbf{Connection to Clustering:}
We can also view this problem in the context of
clustering. Clustering
(\cite{cluster1},\cite{cluster2}, \cite{cluster3})  
is a common problem that arises in the analysis of large data sets. In a typical 
clustering problem, we have a set of $n$ input points in $d$ dimensional space and our goal is to 
partition the points into $k$ clusters. 
There are two ways to define the cluster size (cost): 
\begin{itemize}
 \item the maximum pairwise distance between an arbitrary pair of points in the cluster,
 \item twice the maximum distance between a  point and a chosen centroid.
\end{itemize}

The first one is called as $k$-center clustering 
for diameter cost and the second one is called as $k$-center clustering for radius cost.
In the $k$-center problem, our goal is to minimize the maximum of these distances. Computing 
$k$-center clustering is NP-hard: 
even for $2$ clusters in general Euclidean space (of dimension $d$); and also for 
general number of $k$ clusters  even on a plane.

In this paper,  we assume that the cluster can be of symmetric convex shape also. Given a set $S$ of $n$ points and  
a symmetric convex body $A$ in $\real^d$, we say that the set of points is $(k,A)$-clusterable if 
all the points can be contained in $k$ translated copies of $A$. In the $\textit{promise}$ version
of the problem, for a given proximity parameter $\epsilon~ (\mbox{where~} 0<\epsilon\leq1)$, our goal is to distinguish 
between the cases 
when  $S$ is $(k,A)$-clusterable and when it is $\epsilon$-far from
being $(k,A)$-clusterable. We say that $S$ is $\epsilon$-far from
being  $(k,A)$-clusterable if at least $\epsilon n$ points need to be removed from $S$ in order 
to make it  $(k,A)$-clusterable.

We solve the above problem for $k=1$ with constant number of queries. For $k>1$, we  solve
a \textit{weaker} version of the problem. In order to solve the \textit{promise} version of 
the problem, we have designed a 
randomized algorithm which is generally called as \textit{tester}. 

Our algorithms can also be used to find an \textit{approximately good} clustering. In 
clustering with outliers (anomalies), when we have the ability to ignore
some points as outliers, we present a randomized  
algorithm that takes a constant size sample from input and outputs  radii and centers 
of the  clusters. The benefit of our algorithm is that we construct an \textit{approximate}
representation of such clustering in  time which is independent of the input size.  

The most interesting part of our result is that we initiate
application of Helly-type theorem in property testing in order to
solve the clustering problem.

\subsection{Other related work}
Alon \textit{et al.} \cite{alon} presented testing algorithm for $(k,b)$-\textit{clustering}. A set of points is said to be $(k,b)$-\textit{clusterable} if it can be partitioned into 
$k$ \textit{clusters},  where radius (or diameter) of every cluster is at most $b$. Section $5$ 
of~\cite{alon} presents a testing algorithm for radius cost  under the $L_2$ metric. 
The analysis of this algorithm can be easily generalized to any metric under which each cluster 
is determined by a \textit{simple} convex set (a convex set in
$\real^d$ is called \textit{simple} if its VC-dimension is $O(d)$).   

For testing $1$-center clustering, our result and the result from \cite{alon} give
constant query testing algorithms. Although the two results have
incomparable query complexity (in terms of number of queries depending on  $\epsilon$), 
for testing $k$-center clustering, we give a \textit{weaker}  
query complexity algorithm which works for fixed $k$ and $d$, and for
$\epsilon\in(\epsilon',1]$ where 
$\epsilon'=\epsilon'(k,t)$ (where $t$ is a constant which depends on the shape of the geometric object).
 Alon \textit{et al.} used the sophisticated VC-dimension technique
 while we have used  Helly-type results.

\subsection{Organization of the paper}
In Section~\ref{section:prelim}, we introduce the notations, definitions and state Helly 
and \textit{Helly-type} theorems that are used in this paper. In Section~\ref{section:1piercing}, 
we design the \textit{tester} for $(1,A)$-cluster testing for a given symmetric convex body 
$A$. In Section~\ref{section:kpiercing}, we design the \textit{tester}
 for $(k,G)$-cluster testing for a given geometric object $G$. 
In Section~\ref{section:outliers}, as an application of results from
Sections~\ref{section:1piercing} and ~\ref{section:kpiercing},  we
present an algorithm to find \textit{approximate} clusters with
outliers.

\section{Preliminaries}\label{section:prelim}

\subsection{Definitions}\label{subsection:def}

\noindent \textbf{n-piercing:} \emph{A family of sets is called $n$-pierceable if there exists 
a set $S$ of $n$ points such that each member of the family has a non-empty intersection 
with $S$.}

\noindent \textbf{Homotheticity:} \emph{Let $A$ and $B$ be two geometric bodies in $\real^d$.
$A$ is homothetic to $B$ if there exist  $v\in\real^d$ and  $\lambda >0$ such that 
$A = v +\lambda B $ (where $\lambda$ is called  scaling factor of $B$). 
 In particular, when $\lambda=1$, $A$ is said to be a \textbf{translated copy} of $B$.}

\noindent \textbf{Symmetric convex body:} \emph{A convex body $A$ is called symmetric if it 
is centrally symmetric with respect to the origin, \textit{i.e.}, a point $v\in \real^d$ lies in $A$ if 
and only if its reflection through the origin $-v$ also lies in $A$.
In other words, for every pair of points 
$v_1,v_2\in \real^d$, if  $v_1\in v_2+A$, then $v_2\in v_1+A$ and vice versa.
Circles, ellipses,  $n$-gons (for even $n$) with parallel opposite sides  are examples 
of symmetric convex bodies.}

\subsection{Property Testing}

In property testing, the goal is to query a very small fraction of the input and decide 
whether the input satisfies a certain predetermined property or  is ``far" from satisfying it. Let $x=\{0,1\}^n$ be 
 a given input string. Then, a property testing algorithm, with query complexity $q(|x|)$ and proximity parameter $\epsilon$ for a 
decision problem $L$, is a randomized algorithm that 
makes at 
most $q(|x|)$ queries to $x$ and distinguishes between the following two cases:
\begin{itemize}
 \item if $x$ is in $L$, then the algorithm \textit{Accepts} $x$ with probability at least $\frac{2}{3}$,
 \item if x is $\epsilon$-far from $L$, then the algorithm \textit{Rejects} $x$ with probability at 
 least $\frac{2}{3}$.
\end{itemize}
Here, ``$x$ is $\epsilon$-far from $L$" means that the Hamming distance between $x$ and any 
string in $L$ is at least $\epsilon|x|$.
A property testing algorithm is said to have \textit{one-sided error} if it satisfies the 
stronger condition that the accepting probability for instances $x \in L$ is $1$ instead of $\frac{2}{3}$.

\subsection{Helly's and Fractional Helly's Theorem}
In 1913, Eduard Helly proved the following theorem:
\begin{theorem}\label{theorem:helly}
(Helly's Theorem~\cite{helly}) Given a finite family of convex sets $C_1,C_2,...,C_n$ in $\real^d$ (where $n\geq d+1$) such that if intersection of every 
$d+1$ of  these sets is non-empty, then the whole collection has a non-empty intersection.
\end{theorem}
Katchalski and Liu proved the following result which can be viewed as a fractional version of the
Helly's Theorem.
\begin{theorem}\label{theorem:fractHelly}
 (Fractional Helly's Theorem~\cite{fracHelly}) For every $\alpha~(\mbox{where}~0<\alpha\leq 1)$, there exists $\beta = \beta(d, \alpha)$ with the following property. 
 Let $C_1 , C_2 , ..., C_n$ be convex sets in $\real^d$ (where $n\geq d+1$) and if at least $\alpha {n \choose d+1}$ 
 of the collection of subfamilies of size $d + 1$ has a non-empty intersection, then there exists a point contained in at least $\beta n$ sets. 
\end{theorem}
 Independently, Kalai \cite{UB1} and Eckhoff \cite{UB2} proved that $\beta(d, \alpha)=1-(1-\alpha)^\frac{1}{(d+1)}$. 
A short proof for this upper bound can be found in \cite{shortUB}.

\subsection{\textit{Helly-type} theorem for more than one piercing in convex bodies}
Helly's theorem on intersections of convex sets focuses on  $1$-pierceable families. 
Danzer \textit{et al.} \cite{Danzer82} investigated  the following Helly-type problem :
If $d$ and $m$ are positive integers, what is the least $h = h (d, m)$ such that a family of 
boxes (with parallel edges) in $\real^d$ is $m$-pierceable if each of its $h$-membered 
subfamilies
is $m$-pierceable? Following is the main result of their paper:

\begin{theorem}\label{theorem:danzer}
 \begin{enumerate}
  \item $h(d,1)= 2$ \,\ for all $d~(\mbox{where~}d\geq1)$;
  \item $h(1,m)= m+1$	for all $m$;
  \item $h(d,2)= \begin{cases}
                   3d  \mbox{ for odd } d;\\
                   3d-1 \mbox{ for even } d;
                  \end{cases}$

  \item $h(2,3)= 16$;
  \item $h(d,m)= \infty$ for $d \geq 2, n \geq 3$ and $(d,m)\neq (2,3)$.
 \end{enumerate}
\end{theorem}

Katchalski \textit{et al.} proved  a result for families of homothetic triangles 
in a plane (\cite{Katchalski96}). This result is similar to the intersection property of 
axis parallel boxes in $\real^d$,
studied by Danzer \textit{et al.} This result can also be considered as a Helly-type
theorem for more than one piercing of convex bodies. Theorem~\ref{theorem:triangle}, below, 
presents the main result of their paper.
\begin{theorem}\label{theorem:triangle}
Let $\mathcal{T}$ be a family of homothetic triangles in a plane. If any nine of them can be pierced by 
two points, then all the members of $\mathcal{T}$ can be pierced by two points.
\end{theorem}

This result is  best possible in the following sense:
\begin{itemize}
 \item the bound of \emph{nine} is tight
 \item similar statements do not hold  for homothetic  (or even translated) copies of a symmetric
       convex hexagon
\end{itemize}

\section{Robust Helly for one piercing of symmetric convex body}\label{section:1piercing}
Helly's theorem is a fundamental result in discrete geometry, describing the ways in which convex 
sets  intersect with each other. In our case, we will  focus on those subset of convex  sets 
whose intersection properties behave \emph{symmetric} in certain ways. Observation~\ref{obs:SCB}  
explains this in detail. In order to design the \textit{tester} for $(1,A)$-cluster testing problem, 
we will crucially use this observation,  Helly's  and fractional Helly's theorem.

\begin{obs}\label{obs:SCB}
 Let $A$ be a symmetric convex body in $\real^d$ containing $n$ points, then $n$ translated 
 copies of $A$ centered at these $n$ points have a common intersection. Moreover, a translated 
 copy of $A$ centered at a point in the common intersection contains all these $n$ points.
\end{obs}

\begin{lem}
 Given a set $S$ of  $n$  points in $\real^d$, if every $d+1$ (where $d+1\leq n$) of them are contained
 in (a translated copy of) a symmetric convex body $A$, then all the $n$ points are contained in (a translated copy of) $A$.
\end{lem}

\begin{proof}
 Consider a set $\mathcal{B}$ of  translated copies of $A$ centered at  points in $S$. Since every $d+1$ of the given points 
are contained in (a translated copy of) $A$, by 
 Observation~\ref{obs:SCB}, every $d+1$ elements in $\mathcal{B}$ has a non-empty intersection. 
 By Helly's theorem, all elements in $\mathcal{B}$ have a non-empty intersection. 
 Let $q$  be a point from this intersection. Then $q$ belongs to every element in $\mathcal{B}$ and 
 hence, by Observation~\ref{obs:SCB}, all the centers of the elements in $\mathcal{B}$, \textit{i.e.}, all the $n$ 
 points in $S$, are contained in (a translated copy of) $A$ centered at $q$.
\end{proof}

\begin{lem}\label{lem:symConvFracHelly}
Let $S$ be a set of $n$ points in $\real^d$ (where $n\geq d+1$). If at least $\epsilon n$ (where $0<\epsilon\leq1$) 
points cannot be contained in any translated copy of a symmetric convex body $A$, then at least 
$\epsilon^{d+1}$ fraction of all the $d+1$ size subsets of $S$ (number of such subsets 
is ${n \choose d+1})$  cannot be contained in any translated copy of $A$.
\end{lem}
\begin{proof}
Consider a set $\mathcal{B}$ of  translated copies of $A$ centered at  points in 
 $S$. Now, by fractional Helly's theorem, for every $\alpha~(\mbox{where}~0<\alpha \leq1)$,  
 there exists  $\beta=\beta(d,\alpha)$  such that if at least an $\alpha$ 
 fraction of  $n \choose d+1$ subsets (of size $d+1$) in $\mathcal{B}$ has a non-empty 
 intersection,  then there  exists a point (say $p$) which is contained in at 
 least $\beta$ fraction of elements of  $\mathcal{B}$. 

 Consider a translated copy of $A$ centered at $p$. 
 By Observation~\ref{obs:SCB}, for every $\alpha~(\mbox{where}\linebreak~0<\alpha \leq1)$,  there exists  
 $\beta=\beta(d,\alpha)$  such that if at least an $\alpha$ fraction of  $n \choose d+1$ 
 subsets (of size $d+1$) in $S$ are  contained in  $A$, then at least $\beta n$ 
 points are contained  in  $A$.
 
 Thus, if at least $(1-\beta) n$  points  cannot be contained in $A$, then at 
 least $1-\alpha$ fraction of  $n \choose d+1$  subsets  (of size $d+1$) in $S$ cannot be 
 contained in $A$. (Contrapositive of the above statement.)
 
Since $\beta=1-(1-\alpha)^\frac{1}{(d+1)}$ (\cite{UB1},~\cite{UB2}),  choosing $1-\beta$ as 
$\epsilon$ makes $1-\alpha$ equal to  $\epsilon^{d+1}$, which are the required values of the
parameters.
\end{proof}

\begin{theorem}\label{theorem:testunitball}
 Consider a set of $n$ points in $\real^d$ ($n\geq d+1$) located such that at least $\epsilon n$
 (where $0<\epsilon\leq1$) points cannot be contained in any translated copy of a symmetric convex 
 body $A$. If  we randomly sample $\frac{1}{\epsilon^{d+1}}\ln\frac{1}{\delta}$ (where $0<\delta\leq1$)
  many sets of $d+1$ points, then there exists a set in the sample which cannot be contained in 
  any translated copy of $A$,  with probability at least $1-\delta$.
\end{theorem}

\begin{proof}
 By Lemma~\ref{lem:symConvFracHelly}, if at least $\epsilon n$ points  cannot be contained
 in (any translated copy of) $A$, then at least $\epsilon^{d+1}$ fraction of  $n \choose d+1$ sets 
 (of size $d+1$)  cannot be contained in (any translated copy of) $A$. A  set of $d+1$ points cannot be contained in $A$ with probability $\epsilon^{d+1}$. 
 Hence, the probability that it can be contained in $A$ is  $1-\epsilon^{d+1}$.  Thus, the 
probability that all the sampled sets are contained in $A$ is   
 $ \leq (1-\epsilon^{d+1})^{\frac{1}{\epsilon^{d+1}}\ln\frac{1}{\delta}} \leq e^{-\ln\frac{1}{\delta}}=\delta.$
 \end{proof}

 Algorithm~\ref{algorithm:algo} is a randomized algorithm, \emph{tester}, for $(1,A)$-cluster 
 testing problem.\\ 
 \begin{algorithm}[H]\label{algorithm:algo}
 \SetAlgoLined
 \KwData{A set $S$ of $n$ points in $\real^d$ (input is given as black-box), $0<\delta,\epsilon\leq1$.}
 \KwResult{Returns a set of $d+1$ points, if it exists, which cannot be contained in $A$ 
  or accepts (\textit{i.e.}, all the points can be contained in $A$).}
 \Repeat {$\frac{1}{\epsilon^{d+1}}\ln\frac{1}{\delta}$ many times}
  {
  select a set (say $W$) of $d+1$ points uniformly at random from $S$\\
 \If{$W$  cannot be contained in $A$}
  {return $W$ as witness}
  }
  \If{no witness found }{return \,\  /*\emph{~all the points can be contained in  $A$~}*/}

 \caption{$(1,A)$-cluster testing  in a  symmetric convex body $A$}
\end{algorithm}
This algorithm  has a one sided error, \textit{i.e.},  if all the points can be contained in a 
 symmetric convex body $A$ then it accepts  the input,  else it outputs a witness with probability at least  $1-\delta$. 
 Correctness of the algorithm follows from Theorem~\ref{theorem:testunitball}.
Thus, in the problem of testing $(1,A)$-clustering  for a symmetric convex body $A$, the sample size is independent of the 
input size and hence the property is \emph{testable}. Moreover, the \emph{tester} works for all the possible 
values  of $\epsilon$ (for $0<\epsilon\leq1$).

\section{Robust Helly for more than one piercing of convex bodies}\label{section:kpiercing}

\subsection{Helly-type results for more than one piercing of convex bodies}

The following lemma says that a \emph{``Helly-type"} result is not true for circles even for $2$-piercing. The 
result can be easily generalized for higher dimensions also. (The proof of the following lemma was suggested 
by Prof. Jeff Kahn in a private communication.)

\begin{lem}
 Consider a set of $n$ circles in a plane. For any constant $w~(\mbox{where~}w<n)$, the condition that every  $w$
 circles are pierced at two points is not  sufficient to ensure that all the circles in the 
 set are pierced at two points.
 \end{lem}
 \begin{proof}
We will prove the lemma by construction. Consider a unit circle and  look at any chord in 
it  bounding an arc of angle $(180-\phi)$ degrees (for some very tiny constant $\phi$). If we 
extend out the chord in both directions, it is a straight line which can be viewed as the 
limit of a large circle (that doesn't contain the center of the original unit circle). Consider 
all such large limiting circles obtained by all chords which bound an arc of angle $(180-\phi)$ 
degrees. Then, it is easy to see that there are no two points which pierce all the circles. 
However, for any constant $w$, any $w$ circles are pierced by two points. To see this, take 
any $w$ circles and pick two random antipodal points \footnote{The antipode  of a point on 
the perimeter of a circle is the point which is diametrically opposite to it.}
(say $r_1$ and $r_2$) from the original unit circle. 

The probability that any particular circle 
$C$ is neither pierced at  $r_1$ nor at $r_2$ is $\frac{2\phi}{2\pi}=\frac{\phi}{\pi}$.
Thus, the probability that at least one of the $w$ circles is neither pierced at $r_1$ nor 
at $r_2$ is ${\frac{w\phi}{\pi}}$ (by probability union bound). Hence, the probability that 
all the $w$ circles are  either pierced at  $r_1$ or  
at $r_2$ is $1-\frac{w\phi}{\pi}$.

Thus, for a given $w$, by taking a sufficiently 
small value of $\phi \,\ (\mbox{where~}\phi<\frac{\pi}{w})$, we can make above 
probability high enough. Hence, by  probabilistic arguments, we  prove that there 
exist some two points that pierce any constant number of circles.
\end{proof}

Using arguments similar to the proof of above lemma, it is easy to prove that a 
\emph{``Helly-type"} result for more than one piercing is also not true for a set 
of translated ellipsoids. Katchalski \textit{et al.}~\cite{Katchalski96} 
and Danzer \textit{et al.} \cite{Danzer82} proved a \emph{``Helly-type"} result 
for more than one piercing of triangles and boxes, respectively. According to 
\cite{Katchalski96}, a \emph{``Helly-type"} result for more than one piercing  
is not true for centrally symmetric hexagon (with  parallel opposite edges). 
Similar type of result is true for triangles and pentagons (with pair of parallel 
edges) which are not  symmetric convex bodies. Thus, among symmetric convex bodies 
(spheres, ellipsoids and $n$-gons (for $n\leq6$)), a \emph{``Helly-type"} result for 
more than one piercing is possible only for parallelograms. We have following observation 
regarding the same:
\begin{obs}\label{obs:parallelogram}
 Let $S$ be a set of $n$ points in $\real^d$. 
 If every set of $h$ points (for finite possible values of $h$, see Theorem~\ref{theorem:danzer})  
 in $S$ is contained in $m~(\mbox{where~}m>0)$ translated parallelograms, then all the $n$  points are 
 contained in $m$ translated parallelograms.
\end{obs}
\begin{proof}
 Consider the set $\mathcal{B}$ of translated parallelograms  centered at  points in $S$. 
 Since every set of $h$ points is contained in $m$ translated parallelograms, every 
 $h$-membered subset of $\mathcal{B}$ is  $m$-pierceable. Thus, by Theorem ~\ref{theorem:danzer}, $\mathcal{B}$ 
 is $m$-pierceable. Let  $q_1,q_2,...,q_m$ be $m$ points in these $m$-intersections. Each element 
 of $\mathcal{B}$ contains at least one of the $q_i$'s (where $1 \leq i \leq m$) and, therefore,  $m$ translated parallelograms with centers  as $q_1,q_2,...,q_m$ will 
 contain all the points in $S$ (by Observation~\ref{obs:SCB}). 
\end{proof}

\subsection{Fractional Helly for more than one piercing of convex bodies}
We  now design a \textit{weaker} version of \textit{tester} for $(k,G)$-clustering
(where $G$ is a bounded geometric object and $k>1$). The tester works for
some particular value of $\epsilon\in(\epsilon'(k,t),1]$, where $t$ is some constant that 
depends on the shape of geometric object.

We state the following conjecture for more than one piercing of convex bodies.
\begin{conj}\label{conj:conjecture}
For every $\alpha~(\mbox{where~}0<\alpha\leq 1)$, there exists $\beta = \beta(\alpha,k,d)$ with the following property. 
Let $C_1 , C_2 ,.., C_n$ be convex sets in $R^d$, $n \geq k(d + 1)$, such that at least 
$\alpha.{n \choose k(d+1)}$ 
of the collection of subfamilies of size $k(d + 1)$  are pierced at $k$ points, then at least $\beta n$ 
sets are pierced at  $k$ points. Also, $\beta$ approaches $1$ as $\alpha$ approaches  $1$.
\end{conj}
\begin{lem}\label{lem:conj}
If  Conjecture~\ref{conj:conjecture} is true, then  we   have the following:  Consider a set of $n$ points in $\real^d$ 
(where~$n\geq k(d+1)$). If at least $\epsilon n$ (where $0<\epsilon<1$)
points cannot be contained in any $k$ translated copies of symmetric convex body $A$, then at least
$\gamma(\beta(\epsilon,k,d))$ fraction of $n \choose k(d+1)$ sets cannot be contained in any $k$ translated copies of  $A$.
\end{lem}
\begin{proof}
 Proof of this lemma is similar to the proof of Lemma~\ref{lem:symConvFracHelly}. 
\end{proof}

In the above lemma, $\gamma$ is an appropriately chosen function to compute the value of $1-\alpha$,~\textit{i.e.}, the fraction of $n \choose k(d+1)$ sets which cannot be 
contained in $k$ translated copies of  $A$.

Now, we  prove  a \textit{weaker} version of Conjecture~\ref{conj:conjecture}. We show that 
for bounded geometric objects,  a weaker version of  fractional Helly for more than one piercing is true. We 
use \emph{greedy} approach to prove the same.  
We  prove it for some $\epsilon \in(\epsilon',1]$,
where $\epsilon'=\epsilon'(k,t)$ (where $t$ is a constant that depends on the shape of the geometric 
object). The result is true  only for constant $k$ and $d$.

\begin{lem}\label{lem:kPiercing}
Consider $k$ translated copies of a geometric object $G$ and a set of $n$ points in $\real^d$ 
(for constant $k$ and  $d$). Then there exist $\epsilon'=\epsilon'(k,t)$ 
(where $\epsilon'(k,t)=1-\frac{1}{2(t+1)(k+1)}$, $t$ is a 
constant that depends on the shape of the geometric object) such that
 for all $\epsilon\in(\epsilon',1]$, if at least $\epsilon n$ points cannot be contained 
in  any $k$ translated copies of $G$, then there exist at least  $\Omega(n^{k+1})$ many witnesses of 
 $k+1$ points which cannot be contained in any $k$ translated copies of $G$.
\end{lem}

\begin{proof}
 We say a geometric object $G$ is \emph{best} if it encloses the maximum number of points from the given set of 
 $n$ points. Now, we start with such a  best object. Let us say the best object  contains at least $c_0(1-\epsilon)n$ 
points $(\mbox{where~}0<c_0\leq 1)$.  
Now draw an object, $L_G$, concentric and homothetic with respect to $G$, having a scaling factor of $2+\varepsilon$ 
(for $0<\varepsilon\ll1$, see the definition of Homotheticity in Subsection~\ref{subsection:def} where
$v=0$ and $\lambda=2+\varepsilon$).
The annulus obtained by two concentric objects $G$ and $L_G$ can be filled 
 with constant many  (say $t~(=\kappa^d-1)$, see Lemma~\ref{lem:covering}) 
 translated copies of $G$. Since we started with 
 the best object,  the annulus contains at most $tc_0(1-\epsilon)n$ points. Hence, the number of 
 points which are outside  $L_G$ is at least $\epsilon n -tc_0(1-\epsilon)n=\epsilon_1n$, where $\epsilon_1=\epsilon-tc_0(1-\epsilon)$.
 We throw away all the points in the annulus. Now, we are left with best object that
 containing  at least $c_0(1-\epsilon)n$ points and the remaining space containing at least 
 $\epsilon_1n$ points.
 
 Now, we repeat the above process on $\epsilon_1 n$  points and would keep on repeating it until every point is either  
deleted or contained in some translated copies of  $G$. Thus, total number of points that  we have deleted from 
 annuli is at most~ $t\Sigma_{i\geq0} c_i(1-\epsilon_i)n$ and total number of points that are inside translated
copies of $G$ is at least $\Sigma_{i\geq0} c_i(1-\epsilon_i)n ~~(\mbox{where~} \epsilon_0=\epsilon).$
 
By construction, the total number of points inside translated copies of $G$ and the points that have been 
deleted from annuli is at least $n$. Thus, $$\Sigma_{i\geq0}c_i(1-\epsilon_i)n+t\Sigma_{i\geq0}c_i(1-\epsilon_i)n \geq n~~~(\mbox{where}~\epsilon_0=\epsilon).$$
$$\Sigma_{i\geq0}c_i(1-\epsilon_i)n\geq \frac{n}{t+1}.$$  Let $G_i$ denotes the $i$-th geometric object
and $|G_i|$ denotes  the number of points contained in it. Thus, $$\Sigma_{i\geq0}|G_i|\geq \frac{n}{t+1}.$$
By assumption, $k$ translated copies of  $G$ can contain at most $(1-\epsilon)n$ points. \\
Thus, $|G_i|\leq (1-\epsilon)n$. Since $\epsilon>1-\frac{1}{2(t+1)(k+1)}$, 
$$|G_i|<\frac{n}{2(t+1)(k+1)}.$$
Now, our goal is to make $k+1$ buckets, $S_1,S_2,..,S_{k+1}$, from $G_i$'s such that each bucket
contains  at least $\frac{n}{2(t+1)(k+1)}$ points and at most $\frac{n}{(t+1)(k+1)}$ points. We construct
these buckets by adding points from $G_i$'s until its size become at least $\frac{n}{2(t+1)(k+1)}$.
Since each $|G_i|<\frac{n}{2(t+1)(k+1)}$ and $\Sigma_{i\geq0}|G_i|\geq \frac{n}{t+1}$, this construction
is possible. Thus, for a particular bucket $S_i$,
$$ \frac{n}{2(t+1)(k+1)} \leq |S_i|\leq \frac{n}{(t+1)(k+1)}.$$
 Now, choosing one point from each of the $(k+1)$ buckets gives a set of $k+1$ points as a witness, which cannot be contained in $k$ 
 translated copies of $G$. Thus, there are at least  $ {\left(\frac{1}{2(t+1)(k+1)}\right)}^{k+1}n^{k+1}~(=\Omega(n^{k+1}))$ many witnesses. 
\end{proof}

\begin{theorem}\label{theorem:kPiercing}
 Consider $k$ translated copies of a geometric object $G$ and a set of $n$ points in $\real^d$ 
 (for constant $k$ and  $d$).  Then there exist $\epsilon'=\epsilon'(k,t)$ (where $t$ is a 
constant that depends on the shape of the geometric object) such that
 for all $\epsilon\in(\epsilon',1]$,  at least $\epsilon n$ points cannot be contained in any
 $k$ translated copies of $G$. Now, if  we randomly sample 
 $\frac{1}{c}\ln\frac{1}{\delta}$ (where $0<\delta\leq1$ and $cn^{k+1}$ is the number of witnesses, 
 see Lemma~\ref{lem:kPiercing})  many sets of  size $k+1$, then there exists a set in the 
 sample which cannot be contained in any $k$ translated copies of $G$, with probability at least $1-\delta$.
\end{theorem}

\begin{proof}
  By Lemma~\ref{lem:kPiercing}, if at least $\epsilon n$   points cannot be contained
 in any $k$ translated copies of $G$, then there exist at least  $cn^{k+1}$ many witnesses of $k+1$ points  which cannot be 
 contained in any $k$ translated copies of $G$.
  A  set of $k+1$ points cannot be contained in $k$ translated copies of $G$ with probability 
 $\frac{cn^{k+1}} {n^{k+1}}=c$. 
 Hence, the probability that it can be contained in $k$ translated copies of $G$ is  $1-c$.
  Thus, the probability that all the sampled sets can be contained in $k$ translated copies of $G$ is 
   $ \leq (1-c)^{\frac{1}{c}\ln\frac{1}{\delta}} \leq e^{-\ln\frac{1}{\delta}}=\delta.$
\end{proof}

Similar to \emph{tester} for $(1,A)$-cluster testing problem, we present a \emph{tester} 
(Algorithm~\ref{algorithm:algokPiercing}) for problem $(k,G)$-cluster testing.
If all the points can be contained in $k$ translated copies of 
 $G$ then algorithm accepts the input,  else it outputs a witness with probability at least  $1-\delta$.
 Correctness of the algorithm follows from Theorem~\ref{theorem:kPiercing}.
 Thus, similar to testing $(1,A)$-clustering,  this property is also \emph{testable}. But,
the \emph{tester} only works for constant $k$ and $d$ and for $\epsilon\in(\epsilon',1]$ (see Lemma~\ref{lem:kPiercing}).\\

 \begin{algorithm}[H]\label{algorithm:algokPiercing}
 \SetAlgoLined
 \KwData{A set $S$ of $n$ points in $\real^d$ (input is given as black-box), $0<\delta\leq1$ and 
 $\epsilon\in(\epsilon',1]$.}
 \KwResult{Returns a set of $k+1$ points, if it exists, which cannot be contained in $k$ 
 translated copies of $G$, or accepts (\textit{i.e.}, all the points can be contained in it).}
 \Repeat {$\frac{1}{c}\ln\frac{1}{\delta}$ many times}
  {
  select a set (say $W$) of $k+1$ points uniformly at random from $S$\\
 \If{$W$  cannot be contained in $k$  translated copies of $G$}
  {return $W$ as witness}
  }
  \If{no witness found }{return \,\  /*\emph{~all the points can be contained in $k$  translated copies of $G$~}*/}

 \caption{$(k,G)$-cluster testing  in geometric objects}
\end{algorithm}

\begin{lem}\label{lem:covering}
 Let $G$ be a bounded geometric object. Consider another geometric object $L_G$, concentric and homothetic with respect to $G$,  
 having a scaling factor of $2+\varepsilon$ (for $0<\varepsilon\ll1$, see the definition of Homotheticity in Subsection~\ref{subsection:def}).
 Now, the annulus obtained between two concentric objects $G$ and $L_G$ can be covered by $\kappa^d-1$ translated copies of $G$, 
where $\kappa$ is (ceiling of) the ratio of side length of the smallest $d$-cube 
circumscribing $L_G$ to that of the largest $d$-cube (homothetic \textit{w.r.t.} smallest $d$-cube 
circumscribing $L_G$) inscribing $G$. 
\end{lem}
\begin{proof}
 Let $C_{L_G}$ be the smallest $d$-cube circumscribing $L_G$ and $C_G$ be  
 the largest $d$-cube (homothetic \textit{w.r.t.} $C_{L_G}$) inscribing $G$. Let 
 $C_{L_G}=[0,\kappa]^d$.  Now, consider the $d$-dimensional grid of $C_{L_G}$ obtained 
 by points whose coordinates  are from the set $\{0,1,2,3,..,\kappa\}$. One translated 
 copy of $C_G$ would be require to cover each of the unit $d$-cube from the $d$-dimensional 
 grid, and hence  $\kappa^d$ translated copies of $C_G$ would require to cover $C_{L_G}$. 
 Thus, a covering by $\kappa^d-1$ translated copies of $C_G$ would be required to cover the annulus
 between $C_{L_G}$ and $C_G$.

Now, in order to get a bound for geometric objects, in the above covering,  we can replace 
$C_{L_G}$ by $L_G$ and $C_G$ by $G$. Clearly, the cube covering bound would be an upper bound 
for geometric object covering.
\end{proof}

\section{Application in Clustering with Outliers}\label{section:outliers}
While considering the clustering problem, we mostly assume that data 
is perfectly  clusterable. But a few random points (outliers, noise)
 could be added in the data by an adversary. For example, in the  $k$-center 
clustering, if an adversary adds a  point in the data which is very 
far  from the original set of well clustered points, then in the optimum 
solution that point  becomes center of its own cluster and the remaining points
 are forced to  clustered with  $(k-1)$ centers only. Also, it is even difficult 
to locate when a point  becomes an outlier. For example: consider a set of 
points where we need to  find its optimal $k$-center clustering. Take a point 
from that set and keep  moving it far from the remaining set. Now, it is very 
difficult to locate  correctly at which place that point becomes center of its 
own cluster and  the remaining points are left with $(k-1)$-center clusters.

  In this work, we consider clustering with outliers by  ignoring some fraction
 of points. Thus, in the case when points are perfectly clusterable, ignoring
 some fraction of points does not affect the result too much, and the case
 when outliers are present, the algorithm has the ability to ignore them
 while computing the final clusters.  It may seems that the ability to ignore
 some fraction of points makes the problem easier, but on the contrary it
 does not. Because it has not only to decide which point to include in the
 cluster but also to decide which point to include first. There may be two extreme
 approaches to solve this problem: 1) Decide which points are outliers and run
 the clustering algorithm; 2) Do not ignore any points, and after getting
 final clusters decide which ones are outliers. Unfortunately, neither of
 these two approaches works well.  The first one scales poorly because there are 
 exponentially  many choices, and the second one may significantly change 
 the final  outcome when outliers are indeed present.  This motivates the study 
 of clustering with outliers (see\cite{outlier}).

Theorem~\ref{theorem:testunitball} has an application to $1$-center clustering with outliers. 
More precisely, for  $0<\epsilon,\delta\leq 1$, when we have the ability to ignore  at least
 $\epsilon n$ points as outliers, we present a randomized algorithm which takes a 
constant size sample from input  and correctly  output the radius and center of 
the \textit{approximate} cluster with  probability at least $1-\delta$. 

\begin{algorithm}[H]\label{algorithm:cluster}
 \SetAlgoLined
 \KwData{A set $S$ of $n$ points in $\real^d$ (input is given as black-box), $0<\epsilon,\delta\leq1$.}
          
 \KwResult{Report center and radius of cluster which contain all but at most  $\epsilon n$  points.}
 
 Uniformly and independently, select $m=\frac{d+1}{\epsilon^{d+1}}\ln\frac{1}{\delta}$ points from $S$.\\
 
 Compute minimum enclosing ball containing all the sample points and report its center and radius.\\
  
 \caption{$1$-center clustering with outliers}
\end{algorithm}

\begin{theorem}
 Given a set of $n$ points in $\real^d$ and  $0<\epsilon,\delta\leq1$, Algorithm~\ref{algorithm:cluster} 
 correctly outputs,  with probability at least  $1-\delta$, a ball containing all but 
 at most $\epsilon n$ points in constant   time  by querying a constant  size sample (constant depending
 on $d$ and $\epsilon$). Moreover, if $r_{outlier}$ is the smallest ball containing all but at most $\epsilon n$ points and $r_{min}$ 
 is the smallest  ball containing all the points, then Algorithm~\ref{algorithm:cluster} outputs 
 the radius $r$ such that $r_{outlier}\leq r\leq r_{min}$.
 \end{theorem}
 \begin{proof}
From Theorem~\ref{theorem:testunitball}, if a sample of size $m$ is contained in a ball of radius 
$r$, then this ball would contain all but at most $\epsilon n$ points, with probability 
at least $1-\delta$. And we compute the value of $r$ in step $2$  using Algorithm of~\cite{NIMROD},
which takes  $\theta(m)$ time. Thus, both the sample size and  the running time of the algorithm are constant. 
Clearly, $r_{outlier}\leq r\leq r_{min}$.
\end{proof}

The problem of clustering with outliers can be generalized for $k$-center clustering. 
If Conjecture~\ref{conj:conjecture} is true, then it has an application to $k$-center clustering 
with outliers.  For given $0<\epsilon,\delta\leq 1$, ignoring at least $\epsilon n$
 points as outliers, we present a randomized algorithm which takes a constant size sample 
 from the input and  correctly  output the radii and $k$ centers of the \textit{approximate} 
 clusters with  probability at least $1-\delta$.

\begin{algorithm}[H]\label{algorithm:k-cluster}
 \SetAlgoLined
 \KwData{A set $S$ of $n$ points in $\real^d$ (input is given as black-box), $0<\epsilon,\delta\leq1$.}
 \KwResult{Reports $k$ centers and radii of  clusters which contains all but at most 
          $\epsilon n$  points.}

 Uniformly and independently, select $m=\frac{k(d+1)}{\gamma(\beta(\epsilon,k,d))}\ln\frac{1}{\delta}$ points from $S$.\\
 
 Compute $k$ minimum enclosing balls containing all the sample points and report their centers and radii.\\

 \caption{$k$-center clustering with outliers}
\end{algorithm}

\begin{theorem}\label{theorem:test_k_unitballs}
 Consider a set of $n$ points in $\real^d$. If Conjecture~\ref{conj:conjecture} is true and 
  $0<\epsilon,\delta\leq1$, then with probability at least 
 $1-\delta$, Algorithm~\ref{algorithm:k-cluster} output $k$ balls containing all but at most 
 $\epsilon n$ points in constant   time  by querying a constant  size sample (constant depending
 on $k$, $d$ and $\epsilon$). Moreover, for $1\leq i\leq k$, if $r^{(i)}_{outlier}$ is the radius of the optimal $i$-th cluster 
by ignoring at most $\epsilon n$ points as outliers and $r^{(i)}_{min}$ is the radius of the 
optimal $i$-th cluster when all points are present, then Algorithm~\ref{algorithm:k-cluster} 
outputs the radius  $r^{(i)}$ such that $r^{(i)}_{outlier}\leq r^{(i)}\leq r^{(i)}_{min}$.
 \end{theorem}
 \begin{proof}
If Conjecture~\ref{conj:conjecture} is true, then from Lemma~\ref{lem:conj},  if at least $\epsilon$
 fraction of points cannot be contained in $k$ translated copies of a symmetric convex body $A$, then
 at least $\gamma(\beta(\epsilon,k,d))$ fraction of $n\choose k(d+1)$  sets of size $k(d+1)$ cannot be contained in $k$ translated 
 copies of $A$. Here, $\gamma$ is an appropriately chosen function to compute the value of $1-\alpha$.
 Now, similar to $1$-center clustering with outliers, a sample of size $m=\frac{k(d+1)}{\gamma(\beta(\epsilon,k,d))}\ln\frac{1}{\delta}$
 would be sufficient.
 
 Step $2$ of Algorithm~\ref{algorithm:k-cluster} can be computed in time $O(m^{kd+2})$ using algorithm of Agrawal \textit{et al.}
(see Section $7.1$ of~\cite{kclustering} for details). For relatively smaller values of $d$, we 
can use the Algorithm of~\cite{kclustering1} to get a better running time 
$(m^{O({f(d).k^{1-\frac{1}{d}}})}$, where $f(d)$ is always bounded by $O(d^{\frac{5}{2}}))$.
Thus,  the sample size as well as the running time of the algorithm are constant. Clearly, $r^{(i)}_{outlier}\leq r^{(i)}\leq r^{(i)}_{min}$.
\end{proof}

\section{Conclusion and Open Problems}\label{section:conclusion}
In this paper, we initiated an application of the  Helly (and \textit{Helly-type}) theorem in property 
testing. For $(1,A)$-cluster testing in a symmetric convex body $A$, we showed that  testing can be done with 
constant number of queries and hence proved that the property is \textit{testable}. Alon \textit{et al.} 
\cite{alon} also solved  a similar problem with constant number of queries, using 
combination of  sophisticated arguments in geometric and probabilistic analysis.
 For $1$-center clustering, our result had an  incomparable query complexity in relation 
 (in terms of number of queries depending on  $\epsilon$) with the result of Alon \textit{et al}. We stated a conjecture related to
 fractional \textit{Helly-type} theorem for more than one piercing of convex bodies. Using a 
 greedy approach, we proved a weaker version of the conjecture which we used for testing $(k,G)$-clustering.
 We also gave a characterization of the type of  symmetric convex body for which Helly-type result for more that one piercing would be true.
 Finally, as an application of testing result in clustering with outliers, we showed 
that one can find, with high probability, the \textit{approximate} clusters by querying a
constant size sample.


\bibliographystyle{plain}
\bibliography{reference}

\end{document}